\documentclass[10pt, a4paper]{amsart}
\usepackage[utf8]{inputenc}
\usepackage{amssymb,amsmath,latexsym,amscd,amsfonts,amsthm,mathrsfs}
\usepackage{enumerate, graphicx, color}

\newcommand{\set}[2]{\{#1 : #2 \}}

\newcommand{\re}{\text{\rm Re\,}}
\newcommand{\im}{\text{\rm Im\,}}

\newcommand{\bn}{{\mathbb{N}}}
\newcommand{\br}{{\mathbb{R}}}

\newcommand{\bc}{{\mathbb{C}}}

\newcommand{\ca}{{\mathcal{A}}}
\newcommand{\css}{{\mathcal{S}}}

\newcommand{\csss}{{\mathscr{S}}}
\newcommand{\cm}{{\mathcal{M}}}

\renewcommand{\b}{\beta}
\renewcommand{\l}{\lambda}
\newcommand{\s}{\sigma}

\renewcommand{\r}{\rho}

\newcommand{\p}{\varphi}

\newcommand{\dd}{\Delta}
\renewcommand{\o}{\omega}
\newcommand{\oo}{\Omega}
\newcommand{\g}{\gamma}
\newcommand{\gga}{\Gamma}
\newcommand{\ep}{\varepsilon}
\newcommand{\z}{\zeta}
\newcommand{\pp}{\Phi}

\renewcommand{\sl}{{S_\lambda}}

\newcommand{\mat}{{\mathrm{Mat}_{2,2}(\bc)}}

\newcommand{\nt}{\noindent}

\newcommand{\bsl}{\backslash}

\newcommand{\prt}{\partial}
\newcommand{\ti}{\tilde}

\newcommand{\lp}{\left(}
\newcommand{\rp}{\right)}

\newcommand{\ie}{\emph{i.e.,\ }}

\newcommand{\bg}{{D_{bg}}}
\newcommand{\bgm}{{D_{bg, m}}}
\newcommand{\bgo}{{D_{bg, 0}}}

\newcommand{\R}{\mathbb{R}}
\newcommand{\C}{\mathbb{C}}

\allowdisplaybreaks
\numberwithin{equation}{section}

\newtheorem{theorem}{Theorem}[section]

\newtheorem{lemma}[theorem]{Lemma}
\newtheorem{corollary}[theorem]{Corollary}
\newtheorem{proposition}[theorem]{Proposition}

\newtheorem{remark}[theorem]{Remark}

\begin{document}

\title[Lieb-Thirring inequalities]
{Lieb-Thirring inequalities for an effective Hamiltonian of bilayer graphene}

\author{Ph. Briet}
\address{Center of Theoretical Physics, CNRS, Aix-Marseille University and University of Toulon, Luminy Campus, 163 ave. de Luminy, 13288 Marseille Cedex 9, France}
\email{briet@cpt.univ-mrs.fr}

\author{J.-C. Cuenin}
\address{Institute of Mathematics, University of Munich, Theresien str. 39, D-80333 Munich, Germany}
\email{cuenin@math.lmu.edu}

\address{Department of Mathematical Sciences, Loughborough University, Loughborough, Leicestershire, LE11 3TU United Kingdom}
\email{J.Cuenin@lboro.ac.uk}

\author{L. Golinskii}
\address{B. Verkin Institute for Low Temperature Physics and Engineering of the
National Academy of Sciences of Ukraine, 47 Science ave.,  61103 Kharkiv, Ukraine}
\email{golinskii@ilt.kharkov.ua}

\author{S. Kupin}
\address{IMB, CNRS, Universit\'e de Bordeaux, 351 ave. de la Lib\'eration, 33405 Talence Cedex, France}
\email{skupin@math.u-bordeaux1.fr}



\thanks{The research is partially supported by ANR-18-CE40-0035 grant.}

\begin{abstract}
Combining the methods of Cuenin \cite{jcc2} and Borichev-Golinskii-Kupin \cite{bgk1}, \cite{bgk2}, we obtain the so-called Lieb-Thirring inequalities for non-selfadjoint perturbations of an effective Hamil\-to\-nian for  bilayer graphene.
\end{abstract}

\maketitle

\section*{Introduction and main results}\label{s0}
Since the early 2000-s, a certain amount of attention of the mathematical community has been attracted by the spectral properties of complex (non-selfadjoint) perturbations of model operators from mathematical physics. Among relatively recent papers in this direction,  we quote articles by Demuth-Hansmann-Katriel \cite{dhk}, Frank \cite{rf01}, \cite{rf03}, Frank-Simon \cite{rf02}, Frank-Sabin \cite{fsa1}, Frank-Laptev-Safronov \cite{fls}, Fanelli-Krejčiřík-Vega \cite{fkv1,fkv2}, Mizutani \cite{miz}, Fanelli-Krejčiřík \cite{fk}, Cuenin-Kenig \cite{jcc-ck} and Lee-Seo \cite{lee-seo}, dealing with spectral properties of complex Schrödinger operators. Similar problems for Dirac, fractional Schrödinger and other types of operators were treated in Cuenin-Laptev-Tretter \cite{clat}, Cuenin-Seigl \cite{jcc-ps}, Dubuisson \cite{cd1}, Cuenin \cite{jcc1,jcc3}, Cossetti \cite{cos}, Ibrogimov-Krejčiřík-Laptev \cite{ikl} and Hulko \cite{Hul1,Hul2}. A series of results on spectral analysis of Jacobi matrices can be found in Borichev-Golinskii-Kupin \cite{bgk1,bgk2} and Golinskii-Kupin \cite{gku1}-\cite{gku4}.

In the present article, we are interested in the study of perturbations of bilayer graphene Hamiltonian given by
\begin{equation}\label{e001}
\bgm:=
\begin{bmatrix}
m& 4\prt^2_z\\
4\prt^2_{\bar z}& -m
\end{bmatrix},
\end{equation}
where $m\ge 0$ and
$$
\prt_{z}:=\frac12\lp\prt_{x_1}+i\prt_{x_2}\rp, \quad
\prt_{\bar z}:=\frac12\lp\prt_{x_1}-i\prt_{x_2}\rp.
$$
As usual, we let
$$
L^2(\br^2;\bc^2):=\left\{
f=
\begin{bmatrix}
f_1\\ f_2
\end{bmatrix}:
\|f\|^2_2=\int_{\br^2}|f(x)|^2\, dx<\infty \right\}
$$
to be the standard space of measurable vector-valued functions; here 
$$|f(x)|=(|f_1(x)|^2+|f_2(x)|^2)^{1/2}.
$$
Furthermore, let
$$
H^2(\br^2;\bc^2):=\left\{
f\in L^2(\br^2;\bc^2): \|f\|^2_{H^2}=\int_{\br^2}(1+|\xi|^2)^2|\hat f(\xi)|^2\, d\xi<\infty
\right\}
$$
be the corresponding second order Sobolev space, where $\hat f$ denotes the Fourier transform of a function $f$, see Section \ref{s11} for more notation. It is not difficult to see that
$$
\bgm: H^2(\br^2;\bc^2)\to L^2(\br^2;\bc^2)
$$
is a selfadjoint operator. Since
$$
\bgm^2=(\dd^2+m^2)I_2,
$$
the spectral mapping theorem yields $\s(\bgm):=(-\infty, -m]\cup[m,+\infty)$. The resolvent  set of $\bgm$ is denoted by $\r(\bgm):=\bc\bsl\s(\bgm)$.

Detailed discussion of this and other similar operators from the physical point of view can be found in the book of Katznelson \cite{ka}.

We consider the perturbed operator
\begin{equation}\label{e002}
\bg:=\bgm+V
\end{equation}
with $V\in L^q(\br^2;\mat),\ q\ge 1$. Since the perturbation $V$ is not assumed to be selfadjoint, the operator $\bg$ may be non-selfadjoint as well. For the formal definition of $\bgm+V$ for the class of potentials considered here we allude to the ``factorization method" of Kato \cite{ka1}; see also Gesztesy-Latushkin et al. \cite{glmz}. A version of Weyl's theorem \cite[Theorem 4.5]{glmz} asserts that
\begin{equation}\label{e003}
\s_{ess}(\bg)=\s_{ess}(\bgm)=(-\infty, -m]\cup[m,+\infty),
\end{equation}
where we adopt the convention that $\s_{ess}(\bg):=\s(\bg)\bsl\s_d(\bg)$ and the discrete spectrum $\s_d(D)$ is the set of isolated eigenvalues of $D$ of finite multiplicity.

We shall be interested in distribution properties of the discrete spectrum $\s_d(\bg)$ of the perturbed operator $\bg$. Note that $\s_d(\bg)$ can only accumulate to $\s_{ess}(\bg)$, and we want to find some quantitative characteristics of the rate of accumulation. 

The first step in this direction is to understand better the localization of the discrete spectrum $\s_d(\bg)$. The well-established Birman-Schwinger operator
\begin{equation}\label{e004}
BS_z:=|V|^{1/2}(\bgm-z)^{-1}V^{1/2}, \quad z\in\r(\bgm),
\end{equation}
plays a key role in this problem, see original references by Birman \cite{bi1}, Schwinger \cite{sch1}.
Here, $V(x)=|V(x)|\, U(x)$ is the polar decomposition of the matrix $V(x)$, 
$|V(x)|:=(V(x)^*V(x))^{1/2}$ and $U(x)$ is the corresponding partial isometry. So, $V^{1/2}(x):=|V(x)|^{1/2} U(x)$ for a.~e. $x\in\br^2$. The Birman-Schwinger principle \cite[Theorem 3.2]{glmz} says that $z\in \r(\bgm)$ is an eigenvalue of $\bg$ iff $-1$ is an eigenvalue of the operator $BS_z$. In particular, we have the inclusion
$$
\s_d(\bg)\subset \{z\in\rho(\bgm): \|BS_z\|\ge 1\}.
$$
Laptev-Ferrulli-Safronov \cite[Thm. 1.1]{lfs} obtain the following interesting result.
%
\begin{theorem}[{\cite{lfs}}]\label{t01}
Let $\bgm, \bg$ be as above and $V\in L^q(\br^2;\mat)$, $1<q<4/3$. Then
\begin{enumerate}
\item For $z\in\r(\bgm)$,
\begin{equation}\label{e005}
\|BS_z\|^q=\||V|^{1/2}(\bgm-z)^{-1}V^{1/2}\|^q\le C_q\|V\|^q_q\frac{(|z-m|+|z+m|)^q}{|z^2-m^2|^{q-1/2}}.
\end{equation}
\item In particular,
$$
\s_d(\bg)\subset \left\{z: C_q\|V\|^q_q\frac{(|z-m|+|z+m|)^q}{|z^2-m^2|^{q-1/2}}\ge 1\right\}.
$$
\end{enumerate}
\end{theorem}
Slightly later, the second author \cite[Thm. 1.1, Prop. 2.4]{jcc2} improved the resolvent bound in several respects. First, he showed that the norm of the Birman-Schwinger operator $BS_z$ in the LHS of \eqref{e005} can be taken in an appropriate Schatten-von Neumann class $\css_p,\ p=p(q)$; second, the range of parameter $q$ is extended to $1\le q\le 3/2$. It was observed that these results were optimal in a certain sense. We mention also that \cite[Prop. A.5]{jcc2} addresses more general situations as compared to \cite[Thm. 1.1]{lfs}; in particular, the former is valid for more general differential operators than the bilayer graphene Hamiltonian.

The key to the Lieb-Thirring type inequalities obtained in this article is a claim similar to \cite[Prop. 2.4]{jcc2}. We feel that it is appropriate to give a detailed and a self-contained proof of this result, see Theorem \ref{t1} below. As compared to \cite[Prop. 2.4]{jcc2}, we extend the range of parameter $q$ to $1\le q<\infty$.

\begin{theorem}\label{t1}
Let $\bgm, \bg$ be defined in \eqref{e001}, \eqref{e002}, and $m>0$. For $q\ge 1$ and $\ep>0$, set
\begin{equation}\label{bschw01}
p=p(q,\ep):=\left\{
\begin{array}{ll}
\frac{q}{2-q}+\ep,& 1\le q<4/3,\\
\frac{q}{2-q},& 4/3\le q\le 3/2,\\
2q,& q>3/2.
\end{array}
\right.
\end{equation}
\begin{enumerate}
\item[(I)] Let $1\le q\le 3/2$. There is a \ $C_3>0$ such that for any $A,B\in L^{2q}(\br^2; \mat)$,  one has
\begin{equation}\label{bshw1}
\|A(\bgm-z)^{-1}B\|_{\css_{p}}\le C_3\Phi(z)\|A\|_{2q}\|B\|_{2q}, 
\end{equation}
where 
$$
\Phi(z)=\Phi_q(z):=\frac{|z+m|+|z-m|}{|z^2-m^2|^{q_1}}\,,
$$
$z\in\r(\bgm)$ and $q_1:=1-1/(2q)$. 
\item[(II)] Let $q>3/2$. There is a \ $C_4>0$ such that for any $A,B\in L^{2q}(\br^2;\mat)$, one has
\begin{equation}\label{bshw2}
\|A(\bgm-z)^{-1}B\|_{\css_{p}} \le C_4\Psi(z)\,\|A\|_{2q}\|B\|_{2q}, 
\end{equation}
where
$$
\Psi(z)=\Psi_q(z):=\frac{(|z+m|+|z-m|)^{q_2}}{|z^2-m^2|^{1/q}}\; \frac 1{d^{1-q_2}(z,\s(\bgm))}\,, $$
$z\in\r(\bgm)$ and $q_2:=3/(2q)<1$. Here, $d(z,\s(\bgm))$ is the distance from $z$ to $\s(\bgm)$.
The constants $C_3, C_4$ depend on $m, q,\ep$, but not on $A,B\in L^{2q}(\br^2; \mat)$.
\end{enumerate}
\end{theorem}

The above result along with discussion on Birman-Schwinger operators preceding Theorem \ref{t01} provides the following corollary.
\begin{corollary}\label{c02}\hfill
\begin{enumerate}
\item For $1\le q\le 3/2$ and $V\in L^q(\br^2; \mat)$, we have
$$
\s_d(\bg)\subset \{z : C_3\pp(z)||V||_q\ge 1\}.
$$
In particular, the discrete spectrum $\s_d(\bg)$ is bounded.
\item For $q>3/2$ and $V\in L^q(\br^2; \mat)$, we have
$$
\s_d(\bg)\subset \{z : C_4\Psi(z)||V||_q\ge 1\}.
$$
\end{enumerate}
\end{corollary}

Theorem \ref{t1} combined with techniques developed in Borichev-Golinskii-Kupin \cite{bgk1}, \cite{bgk2} implies the following result.

\begin{theorem}\label{t-lthin1}
Let $\bgm, \bg$ be defined in \eqref{e001}, \eqref{e002}, and $m>0$.  For $q>1$ and $\ep>0$, set
\begin{equation}\label{e006}
\b=\b(q,\ep):=\left\{
                  \begin{array}{ll}
                    \frac{4q-5}{2(2-q)}+\frac{2q-1}{2q}\ep, & 1<q<\frac43,\\
                    \frac{4q-5}{2(2-q)}, & \frac43\le q\le\frac32.
                  \end{array}
                \right.
\end{equation}
Assume that $\|V\|_q\le 1$. Then the Lieb--Thirring inequalities for the discrete spectrum $\s_d(\bg)$ hold:
\begin{enumerate}
\item[(I)] for $1\le q\le 3/2$,
\begin{equation}\label{lth01}
\sum_{\z\in\s_d(\bg)} d^{1+\ep}(\z,\s(\bgm))\,|\z^2-m^2|^\b \le C_5\|V\|_q,
\end{equation}
\item[(II)] for $q>3/2$,
\begin{equation}\label{lth02}
\sum_{\z\in\s_d(\bg)} \frac{|\z|^{2q+1+\ep}d^{2q-2+\ep}(\z,\s(\bgm))\,|\z^2-m^2|}{(1+|\z|)^{2q+1+\ep}} \le C_6\|V\|_q.
\end{equation}
\end{enumerate}
The constants $C_5, C_6$ depend on $m, q,\ep$, but not on $V\in L^q(\br^2;\mat)$.
\end{theorem}

The counterparts of the above theorems for the case $m=0$ are given below. Their proofs are similar to Theorems \ref{t1}, \ref{t-lthin1}, and therefore they are omitted.

\begin{theorem}\label{t11}
Let $\bgo, \bg$ be given by  \eqref{e001}, \eqref{e002} and $z\in\r(\bgo):=\bc\bsl\br$.\, Take an $\ep>0$ and put $p=p(q,\ep)$ as in \eqref{bschw01}.
\begin{enumerate}
\item[(I)] Let $1\le q\le 3/2$. There is a \ $C'_3>0$ such that for any $A,B\in L^{2q}(\br^2; \mat)$,  one has
\begin{equation}\label{bshw10}
\|A(\bgo-z)^{-1}B\|_{\css_{p}}\le C'_3\, |z|^{-(1-\frac 1q)}\|A\|_{2q}\|B\|_{2q}.
\end{equation}
\item[(II)] Let $q>3/2$. There is a \ $C'_4>0$ such that for any $A,B\in L^{2q}(\br^2;\mat)$, one has
\begin{equation}\label{bshw20}
\|A(\bgo-z)^{-1}B\|_{\css_{p}} \le C'_4\, |z|^{-\frac 1{2q}}|\im z|^{-(1-\frac 3{2q})}\|A\|_{2q}\|B\|_{2q}, 
\end{equation}
Above,  $|\im z|=d(z,\br)$ is the distance from $z$ to the real line $\br$.
The constants $C'_3, C'_4$ depend on $q,\ep$, but not on $A,B\in L^{2q}(\br^2; \mat)$.
\end{enumerate}
\end{theorem}

Similarly to Corollary \ref{c02}, we can decribe the regions containg the discrete spectrum $\s_d(\bg)$ for $m=0$. In particular, the set is bounded for $1\le q\le 3/2$ and $V\in L^q(\br^2; \mat)$.

\begin{theorem}\label{t12}
Let $\bgo, \bg$ be defined as above.  Let $q>1$ and $\ep>0$ be small enough. 
Assume that $\|V\|_q\le 1$. Then the Lieb--Thirring inequalities for the discrete spectrum $\s_d(\bg)$ hold:
\begin{enumerate}
\item[(I)] for $1\le q\le 3/2$,
\begin{equation}\label{lth01-}
\sum_{\z\in\s_d(\bg)} |\im\z| ^{1+\ep}\le C'_5\|V\|_q,
\end{equation}
\item[(II)] for $q>3/2$,
\begin{equation}\label{lth02-}
\sum_{\z\in\s_d(\bg)} \frac{|\im\z|^{2-\frac 3{2q}+\ep}}{(1+|\z|)^{1-\frac 3{2q}+2\ep}}
\le C'_6\|V\|_q.
\end{equation}
\end{enumerate}
The constants $C'_5, C'_6$ depend on $q,\ep$, but not on $V\in L^q(\br^2;\mat)$.
\end{theorem}

\begin{remark}\label{r1}\hfill
\begin{enumerate}
\item In order to prove the above theorems we need the $\css_p$-norm of the Birman-Schwinger operator $\|V_2(\bgm-z(iy))^{-1}V_1\|_{\css_p}$ to go to zero when $y\to +\infty$, see \eqref{bshw11}. For this reason inequality \eqref{lth01} is obtained for $1<q\le 3/2$, even though the case $q=1$ is treated in Theorem \ref{t1}.
\item The assumption $\|V\|_q\le 1$ does not mean that the perturbation is small. Theorem \ref{t-lthin1} holds uniformly over any bounded in $L^q$ set of potentials $V$,
i.e., $1$ can be replaced with a constant $C(q,m,\ep)$.
\end{enumerate}
\end{remark}

The paper is organized in the following manner. We start Section \ref{s1} recalling some basic facts and notation on differential operators. The second part of Section \ref{s1} is devoted to the proof of Theorem \ref{t1}. The proof of Theorem \ref{t-lthin1} is in Section \ref{s2}. Section \ref{s4} is an appendix containing results on interpolation between $\css_p$-spaces and the Kato-Selier-Simon lemma.

The space of infinitely differentiable functions on $\br^2$ is denoted by $C^\infty(\br^2)$; $C^\infty_0(\br^2)$ are infinitely differentiable functions with compact support. The notation $L^p(\br^2),\ 1\le p\le\infty$, stays for the familiar space of $p$-summable measurable functions. $L^\infty_0(\br^2)$ refers also to functions from $L^\infty(\br^2)$ with compact support. Meaningful constants are written as $C_j, C'_j,\; j=0,1,\dots$; technical constants are denoted by $c, C$, and they change from one relation to another.

\section{Resolvent bounds for the bilayer graphene Hamiltonian}\label{s1}
\subsection{Fourier transforms}\label{s11}
The purpose of this subsection is to fix some notation and recall some basic properties of the Fourier transformation. For this purpose we temporarily consider the case of arbitrary dimension $n$. At the end of the subsection we will compute Fourier transforms of some tempered distributions (homogeneous distributions and surface-carried measures) that will play an important role in the next subsection. We refer to Hörmander \cite{ho1}, Sogge \cite{so} for more details on the subject. 

The Fourier transform of a function $f\in L^1(\br^n)$ is defined as 
$$
(\mathcal{F}f)(\xi):=\hat f(\xi):= \int_{\br^n} f(x)e^{-i x\cdot \xi}\,dx.
$$
Let $\csss=\csss(\br^n)$ denote the Schwartz space, \ie the space of rapidly decreasing smooth functions on $\br^n$.
The Fourier transformation is an isomorphism $\mathcal{F}:\csss\to\csss$, and its inverse is furnished by the Fourier inversion formula,
$$
f(x)=\frac 1{(2\pi)^n}\int_{\br^n} \hat f(\xi)e^{ix\cdot \xi}\,dx.
$$
We use the standard notation $\check f:=\mathcal{F}^{-1}f$. Hence, $\mathcal{F}$ may be extended to the dual space $\csss'$, the space of tempered distributions, by setting $\hat{u}(\phi)=u(\hat{\phi})$ for $u\in\csss'$, $\phi\in\csss$. Moreover, Plancherel's formula,
\begin{align}\label{Plancherel's formula}
\|\hat f\|_2=(2\pi)^{n/2}\|f\|_2,\quad f\in\csss,
\end{align}
gives rise to a continuous extension $\mathcal{F}:L^2(\R^n)\to L^2(\R^n)$.

Let $D=\nabla$ be a formal differential operator. The Fourier multiplier $m(D):\csss\to \csss'$ associated to a tempered distribution $m\in \csss'$ is the operator
\begin{align*}
m(D)f:=\mathcal{F}^{-1}(m\hat{f}),\quad f\in \csss,
\end{align*}
and \eqref{Plancherel's formula} shows that $m$ is bounded on $L^2(\R^n)$ if and only if $m\in L^{\infty}(\R^n)$, and $\|m(D)\|=\|m\|_{\infty}$. We also have 
\begin{equation}\label{e1}
(m(D)\p)(x)=\check m\, *\p=\int_{\br^n}\check m(x-y)\p(y)\, dy, \quad \p\in \csss,
\end{equation}
with the understanding that $*:\csss'\times\csss\to\csss'$ is the convolution between a Schwartz function and a tempered distribution. The second identity in \eqref{e1} is in general only formal, but it is rigorous if $\check m$ is a regular tempered distribution. To simplify notation, the expression $(m(D))(x)$, refers to the convolution kernel $\check m(x)$ of the integral operator in \eqref{e1}. 

Consider now a smooth real-valued function $\rho$ which we think of as (a normalized power of) a Hamiltonian. Then, for $\l\in\br$, we define the level sets of $\r$ (\ie the sets of constant energy) as 
\begin{equation}\label{e2}
S_\l:=\r^{-1}(\l)=\{\xi\in\br^n: \r(\xi)=\l\}.
\end{equation}
These sets play a crucial role in scattering theory, see e.g. H\"ormander \cite[Ch. XIV]{ho2}. In the present paper the main feature of $S_{\lambda}$ is its nowhere vanishing Gaussian curvature. To ensure that $S_{\lambda}$ is in fact a manifold (a curve) we make the assumption that $\rho$ is normalized such that $|\nabla\rho|=1$ on $S_\l$. In the following we will only deal with\footnote{The fact that $\xi\mapsto|\xi|$ is not smooth at $\xi=0$ is irrelevant for our purposes since (by homogeneity) we will only need smoothness in a neighborhood of the unit sphere $S_1=\{\xi: |\xi|=1\}$.} $\rho(\xi)=|\xi|$, in which case $S_{\lambda}$ is just the sphere of radius $\lambda$. 
Let $d\s_{S_\l}$ be the canonical surface measure on $S_\l$. As usual, $L^2(d\s_\sl)$ is the space of measurable square-summable functions on $S_\l$. 
%
The Fourier restriction operator for $S_{\lambda}$ is defined by 
$$
R(\l)\p:=\hat\p\big|_\sl,\quad \p\in \csss(\br^n).
$$
Its formal adjoint (the Fourier extension operator) is given by
$$
R(\l)^*\p=\widehat{\p\, d\s_{S_\l}},\quad \p\in \csss(\br^n).
$$ 
Here, the Fourier transform of the measure $\p\, d\s_{S_\l}$ is defined as
\begin{align*}
\widehat{\p\, d\s_{S_\l}}(x)=\int_{\R^n}e^{-ix\cdot\xi}\p(\xi)d\s_{S_\l}(\xi).
\end{align*}
The multiplier corresponding to the function $\xi\mapsto |\xi|$ is denoted by $\sqrt{-\Delta}$. Denote by $E_{\sqrt{-\Delta}}(\l)$ the (operator-valued) spectral measure associated to this operator, viewed as an unbounded selfadjoint operator on $L^2(\R^n)$. Since its spectrum is absolutely continuous we may write $dE_{\sqrt{-\Delta}}(\l)=\frac{dE_{\sqrt{-\Delta}}(\l)}{d\lambda}d\lambda$, where the convolution kernel of the density is given by
\begin{align*}
\frac{dE_{\sqrt{-\Delta}}(\l)}{d\lambda}(x-y)=(2\pi)^{-n}\int_{|\xi|=\lambda}e^{i(x-y)\cdot\xi}d\sigma_\sl(\xi).
\end{align*}
By a change of variables $\xi=\l\xi', \ |\xi'|=1$, we see that
\begin{align}\label{e11bis}
\frac{dE_{\sqrt{-\Delta}}(\l)}{d\lambda}=\frac{\lambda^{n-1}}{(2\pi)^n}R(\lambda)^*R(\lambda),
\end{align}
where $R(\l)$ is the restriction operator discussed above. It is also plain that 
$$
R(\lambda)f=\lambda^{-n}R(1)(f(\lambda^{-1}\cdot)).
$$
Define 
$$
\chi_{+}^{w}(\tau):=\mathbf{1}_{[0,\infty)}(\tau)\tau^{w}/\Gamma(w+1),\quad  w\in\C,
$$ 
where $\Gamma$ is the usual Gamma function.

\begin{lemma}\label{lemma FT1}
Let $z,\zeta\in\C$, $\im z>0$. The one-dimensional inverse Fourier transform of the function $\eta_{z,\zeta}(x):=(x-z)^{-\zeta}$, $x\in\R$, is given by 
\begin{align}\label{formula FT1}
\check{\eta_{z,\zeta}}(\tau)=e^{i(\pi\zeta/2+z\tau)}\chi_+^{\zeta-1}(\tau).
\end{align}
\end{lemma}

\begin{proof}
After a change of variables, this follows immediately by applying the inverse Fourier transformation to the following identity (see \cite{ho1}, specifically the explanation after Example 7.1.17)
\begin{align*}
\mathcal{F}\left(x\mapsto e^{-\epsilon x}\chi_+^z(x)\right)(\xi)=e^{-i\pi(z+1)/2}(\xi-i\epsilon)^{-z-1},\quad \epsilon>0,\quad z\in\C.
\end{align*}
\end{proof}
\begin{lemma}\label{lemma FT2}
Let $\beta\in C_0^{\infty}(\R^n)$ and let $S_1$ be the unit sphere in $\R^n$. Then the inverse Fourier transform of the surface measure $d\mu:=\beta \,d\s_{S_1}$ admits the representation
\begin{align*}
\check{d\mu}(x)=\sum_{\pm}e^{\pm i|x|}a_{\pm}(|x|):=e^{i|x|}a_+(|x|)-e^{-i|x|}a_-(|x|),
\end{align*}
where $a_{\pm}\in C^{\infty}(\R_+)$ satisfy the symbol bounds
\begin{align}\label{symbol bounds}
|\partial^ka_{\pm}(s)|\leq C_{k\pm}(1+|s|)^{-\frac{n-1}{2}-k}.
\end{align}
\end{lemma}

\begin{proof}
This is a special case of \cite[Theorem 1.2.1]{so}.
\end{proof}

\begin{lemma}\label{lemma FT3}
Let $\chi\in C_0^{\infty}(\R^n)$ be supported in the annulus $\{1/2\leq |\xi|\leq 3/2\}$, and  $S=\{\z: a\le\re\z\le b\}$ be a vertical strip in $\C$. Then   
\begin{align*}
\left|\int_{\R^n}e^{-ix\cdot\xi}\frac{\chi(\xi)}{(|\xi|-z)^{\zeta}}\, d \xi\right|\leq Ce^{\pi^2|\im\zeta|^2}(1+|x|)^{-\frac{n+1}{2}+\re\zeta},\quad \zeta\in S,\quad |z|=1,
\end{align*}
where the constant depends on $a,b$ and finitely many derivatives of $\chi$, but is independent of \ $\zeta,z$.
\end{lemma}

\begin{proof}
It suffices to prove this for $|x|>1$ since the case $|x|\leq 1$ is trivial. Writing the integral in polar coordinates and using Lemma \ref{lemma FT2} we find that
\begin{align*}
\int_{\R^n}e^{-ix\cdot\xi}\frac{\chi(\xi)}{(|\xi|-z)^{\zeta}}\, d \xi
=\sum_{\pm}\int_{-\infty}^{\infty}e^{\pm ir|x|}\frac{r^{n-1}a_{\pm}(r|x|)}{(r-z)^{\zeta}}\,dr,
\end{align*}
where the function $r\mapsto r^{n-1}a_{\pm}(r|x|)$ is supported in a neighborhood of $r=1$ and it satisfies 
\begin{align*}
|r^{n-1}a_{\pm}(r|x|)|\leq C(1+|x|)^{-\frac{n-1}{2}}
\end{align*}
for any fixed Schwartz norm $|\cdot|$. Hence, by Lemma \ref{lemma FT2} again, its inverse Fourier transform is bounded by
\begin{align*}
|\mathcal{F}^{-1}\left(r\mapsto r^{n-1}a_{\pm}(r|x|)\right)(\tau)|\leq C_N(1+|\tau|)^{-N}(1+|x|)^{-\frac{n-1}{2}}
\end{align*}
for any $N>0$. The convolution theorem and Lemma \ref{lemma FT1} yield
\begin{equation*}
\begin{split}
&\left|\int_{-\infty}^{\infty}e^{\pm ir|x|}\frac{r^{n-1}a_{\pm}(r|x|)}{(r-z)^{\zeta}}\,dr\right|\\
&\quad \leq C_Ne^{\pi|\im \zeta|}(1+|x|)^{-\frac{n-1}{2}}\int_{-\infty}^{\infty}(1+|\tau-|x||)^{-N} \chi_+^{\re\zeta-1}(\tau)\, d\tau\\
&\quad \leq Ce^{\pi|\im \zeta|}|\Gamma(\zeta)^{-1}|(1+|x|)^{-\frac{n+1}{2}+\re\zeta}.
\end{split}
\end{equation*}
The claim now follows from the estimate $|\Gamma(\zeta)^{-1}|\leq Ce^{\pi^2|\zeta|^2}$; see e.g.\ formula (11.21) in Muscalu-Schlag \cite{ms}.
\end{proof}

\subsection{Resolvent bounds in $\css_p$-norm for bilayer graphene}\label{s13}
We now return to the case $n=2$ and the bilayer Hamiltonian. The coming bound is a special case of \cite[Lemma A.6]{jcc2}. It is crucial for coming resolvent estimates.

In the following, we fix a function $\chi\in C_0^{\infty}(\br^2)$ supported in the annulus $\{1/2\leq |\xi|\leq 3/2\}$ such that, in addition, $\chi(\xi)=1$ for $3/4\le |\xi|\le 5/4$.

\begin{proposition}
\label{p1} Let $1\le a\le 3/2, t\in\br$, and $z\not\in\br_+$. There exists a constant $C'_1>0$ (depending on $\chi$ only) such that
\begin{equation}\label{e501}
|\chi(D)(\dd^2-z)^{-(a+it)}(x)|\le \frac{C'_1e^{\pi^2 t^2}}{(1+|x|)^{3/2-a}}, \quad x\in\br^2,\quad |z|=1.
\end{equation}
\end{proposition}

\begin{proof}
Set $z^{1/4}=|z|^{1/4}e^{(i\mathrm{Arg}\, z)/4}$. Clearly the 4-th power complex roots of $z$ are given by $\{i^m z^{1/4}\},\ m=0,1,2,3$. Without loss of generality,  we suppose that $m=0$ and $|\mathrm{Arg}\, z|\le\pi$, or $|\mathrm{Arg}\, z^{1/4}|\le\pi/4$, the other cases being analogous. Writing
\begin{align*}
(|\xi|^4-z)=(|\xi|-z^{1/4})\lp\prod_{k=1}^3(|\xi|-i^kz^{1/4})\rp
\end{align*} 
and absorbing the second factor into $\chi$, we see that it suffices to prove 
\begin{align*}
\int_{\R^n}e^{ix\cdot\xi}\frac{\ti\chi(\xi;a;t)}{(|\xi|-z^{1/4})^{a+i t}}\, d \xi\le \frac{Ce^{\pi^2 t^2}}{(1+|x|)^{3/2-a}},
\end{align*}
whenever $\ti\chi(\xi;a,t)$ satisfies the bounds
\begin{align*}
\sum_{|\alpha|\leq N}\|\partial_{\xi}^{\alpha}\ti\chi(\cdot;a,t)\|_{\infty}\leq C_N e^{2\pi|t|}
\end{align*}
for a fixed, sufficiently large $N>0$. This follows directly from Lemma \ref{lemma FT3}.
\end{proof}

\begin{remark}
In view of the identity 
\begin{align*}
\frac{1}{|\xi|^2-z^{1/2}}-\frac{1}{|\xi|^2+z^{1/2}}=\frac{2z^{1/2}}{|\xi|^4-z},
\end{align*}
inequality \eqref{e501} also follows from a two-dimensional version of  estimates (2.23) and (2.25) in Kenig-Ruiz-Sogge \cite{krs}; see also (44) in Frank-Sabin \cite{fsa1}. To keep the article self-contained, we provided the above proof which rests only on the stationary phase method (Lemma \ref{lemma FT2}) and formula \eqref{formula FT1}.
\end{remark}

\begin{proposition}\label{p2} Fix an $\ep>0$ and set the function $\chi$ as above. For $q\ge 1$, let
\begin{equation}\label{e6}
p=p(q,\ep):=
\left
\{
\begin{array}{ll}
\frac{q}{2-q}+\ep,& 1\le q<4/3,\\
\frac{q}{2-q},& 4/3\le q\le 3/2,\\
2q,& q>3/2.
\end{array}\right.
\end{equation}
For $A,B\in L^{2q}(\br^2)$, the following bounds hold true:
\begin{enumerate}
\item[(I)] for $1\le q\le 3/2$,
\begin{equation}\label{e51}
\|A\chi(D)(\dd^2-z)^{-1}B\|_{\css_p}\le C_7\; \|A\|_{2q}\|B\|_{2q},\quad |z|=1;
\end{equation}
\item[(II)] for  $q>3/2$
\begin{equation}\label{e52}
\|A\chi(D)(\dd^2-z)^{-1}B\|_{\css_p}\le \frac{C_8}{d(z,\br_+)^{1-3/(2q)}}\; \|A\|_{2q}\|B\|_{2q},\quad |z|=1.
\end{equation}
\end{enumerate}
Here, $C_j=C_j(q,\ep),\; j=7,8$, are independent of $A, B$ and $z$.
\end{proposition}
%

\begin{proof}
The proof relies heavily on interpolation between Schatten-von Neumann classes $\css_p, \ p\ge 1$, presented in Section \ref{s4}. It is convenient to separate part (I) of the proposition in two cases: \emph{Case I.1} for $1\le q<4/3$ and \emph{Case I.2} for $4/3\le q\le 3/2$. We begin with the proof of \emph{Case I.2}.

\medskip\nt
{\bf Case I.2: $4/3\le q\le 3/2$.}\ Without loss of generality we may assume that $A>0$ and $B>0$. At the moment, we suppose also that $A, B\in L^{2q}(\br^2)\cap L^\infty_0(\br^2)$.
We wish to apply Corollary \ref{c01} to the analytic family of operators given by 
$$
T_\z:=A^\z\chi(D)(\dd^2-z)^{-\z}B^\z
$$
on the strip $S=S_{0,a_0}:=\{\z: 0\le\re z\le a_0\}$, with $1\le a_0\le 3/2$. Here, $\z=a+it, \ 0\le a\le a_0$, and $t\in\br$.

We start by checking assumptions of Corollary \ref{c01}, see also Theorem \ref{t43}. For arbitrary $f,g\in L^2(\br^2)$ we have, by Plancherel's identity, 
\begin{align*}
(T_\z f,g)=\int_{\R^2}\chi(\xi)(|\xi|^4-z)^{-\z}
\widehat{B^\z f}(\xi)\overline{\widehat{A^\z g}(\xi)}d\xi,
\end{align*}
which shows that $\zeta\mapsto (T_\z f,g)$ is analytic in $S$.
By Cauchy-Schwarz inequality,
\begin{eqnarray}\label{e61}
|(T_\z f,g)|\le \|\chi\|_{\infty}\|(|\cdot|^4-z)^{-\z}\|_{\infty}\|B^\z f\|_2\, \|A^\z g\|_2.\nonumber
\end{eqnarray}
Since $|\arg(|\xi|^4-z)|\le 2\pi$, we have that
\begin{align*}
|(|\xi|^4-z)^{-\z}|
&=\big |\exp(- (a+it)\,(\log\|\xi|^4-z|+i\arg(|\xi|^4-z))\big|\\
&\le|(|\xi|^4-z)|^{-a}\exp(2\pi |t|). 
\end{align*}
Observe  that $a$ varies over a compact interval and $z$ is fixed. Putting  all this together, we obtain that
$$
|(T_\z f,g)|\le  Ce^{2\pi|t|}||\chi||_\infty \|A\|_{\infty}^{a}\|B\|_{\infty}^{a}\|f\|_2\|g\|_2,\quad\z=a+it,
$$
showing that \eqref{e20} is satisfied. It also yields that
\begin{equation}\label{e7}
\|T_{\z}\|_{\css_\infty}
\le  Ce^{2\pi|\im\z|}
\end{equation}
for $\re\z=0$. Note that $T_{\z}$ is compact since we have the Hilbert-Schmidt bound
\begin{align*}
\|T_{\z}\|_{\css_2}^2=&\int_{\br^2_x}\int_{\br^2_y}|A^{\z}(x)|^2|\mathcal{F}\left(\chi(|\cdot|^4-z)^{-\zeta}\right)(x-y)|^2|A^{\z}(x)|^2 dxdy\\
&\le e^{4\pi|\im \z|}
\|\chi(|\cdot|^4-z)^{-\re\z}\|_1^2\|A\|^{2\re\z}_2\|B\|^{2\re\z}_2,
\end{align*}
and the right hand side is finite by the assumption that $A,B\in L^\infty_0(\br^2)$.

On the vertical line $\{\z: \re\z=a_0\}$, Proposition \ref{p1} and Hardy-Littlewood-Sobolev inequality (see Lieb-Loss \cite[Sect. 4.3]{lie1}) yield that
\begin{equation*}
\begin{split}
\|T_{a_0+it}\|^2_{\css_2}
&\le \int_{\br^2_x}\int_{\br^2_y}|\chi(D)(\dd^2-z)^{-(a_0+it)}(x-y)|^2|A(x)|^{2a_0}|B(y)|^{2a_0}\, dxdy\\
&\le C e^{2\pi^2t^2}\int_{\br^2_x}\int_{\br^2_y}\frac1{|x-y|^{3-2a_0}} |A(x)|^{2a_0}|B(y)|^{2a_0}\, dxdy\\
&\le Ce^{2\pi^2t^2} \|A|^{2a_0}\|_s\||B|^{2a_0}\|_s,
\end{split}
\end{equation*}
where $2/s+(3-2a_0)/2=2$, or $s=4/(1+2a_0)$. In particular,
$$
\||A|^{2a_0}\|_s=\|A\|^{2a_0}_{8a_0/(1+2a_0)},
$$
the same equality holding for $\||B|^{2a_0}\|_s$. Hence, gathering the above computations, we arrive at the bound 
\begin{equation}\label{e8}
\|T_\z\|_{\css_2}\le Ce^{\pi^2|\im\zeta|^2} \|A\|^{a_0}_{8a_0/(1+2a_0)}\,\|B\|^{a_0}_{8a_0/(1+2a_0)}
\quad \mbox{for}\quad \re\z=a_0.
\end{equation}

We recall now Corollary \ref{c01} (see also Theorem \ref{t43}) with parameters chosen as
$$
\z:=1, \quad 1=\g\cdot a_0+(1-\g)\cdot 0, \quad \frac 1{s_\g}=\frac\g 2+\frac{(1-\g)}\infty=\frac\g 2,
$$
to interpolate between \eqref{e7} and \eqref{e8}. Solving first for $\gamma$ and then for $s_{\gamma}$ yields $\g=1/a_0$ and $s_\g=2a_0$. Corollary \ref{c01} then implies that
$$
\|A\chi(D)(\dd^2-z)^{-1}B\|_{\css_{2a_0}}\le C_7\; \|A\|_{8a_0/(1+2a_0)}\|B\|_{8a_0/(1+2a_0)},
$$
which is exactly \eqref{e51} with $4/3\le q\le 3/2$ if one puts $2q=8a_0/(1+2a_0)$.

To sum up, we proved \eqref{e51} for $4/3\le q\le 3/2$ and $A,B\in L^{2q}(\br^2)\cap L^\infty_0(\br^2)$. It remains to get rid of the assumption that $A,B\in L^\infty_0(\br^2)$. The proof relies essentially on the fact that the constant $C_7$ from \eqref{e51} \emph{does not depend} on $A, B$. We proceed by a limiting argument. Let $A, B\in L^{2q}(\br^2)$. For $n\in\bn$, define
$$
E_n=\{x\in\br^2:|x|+|A(x)|+|B(x)|\le n\}
$$
and set the ``truncations'' of $A, B$ to be
$$
A_n=A\mathbf{1}_{E_n},\quad B_n=B\mathbf{1}_{E_n}.
$$
Let $P_n:L^2(\br^2)\to L^2(\br^2)$ be the corresponding  orthogonal projection 
$$
P_n f=\mathbf{1}_{E_n} f, \quad f\in L^2(\br^2).
$$
The elementary properties of $L^{2q}$-integrable functions yield that for any $f\in L^2(\br^2)$, we have
$$
\lim_{n\to+\infty}\|P_nf-f\|_2=0.
$$
Recalling \cite[Thm. 5.2]{gk1} and inequality \eqref{e51} for functions from $L^{2q}(\br^2)\cap L^\infty_0(\br^2)$, we obtain
\begin{equation*}
\begin{split}
&\|A\chi(D)(\dd^2-z)^{-1}B\|_{\css_p}=\sup_n \|P_n\big(A\chi(D)(\dd^2-z)^{-1}B\big)P_n\|_{\css_p}\\
&=\sup_n \|A_n\chi(D)(\dd^2-z)^{-1}B_n\|_{\css_p}\le C_7 \|A_n\|_{2q}\|B_n\|_{2q}\le C_7 \|A\|_{2q}\|B\|_{2q}.
\end{split}
\end{equation*}
\emph{Case I.2} follows.

\medskip\nt
{\bf Case II: $q>3/2$.} 
As before, we may assume without loss of generality that $A,B\in L^{2q}(\br^2)\cap L^\infty_0(\br^2)$, and that $A,B>0$. 

Let $S:=S_{0,a_0}:=\{a+it: 0\le a\le a_0=2q/3, \ t \in \br\}$. Notice that $q>3/2$ implies that $a_0=2q/3>1$. Consider the analytic family of operators 
$$
T_\z=A^\z\chi(D)(\dd^2-z)^{-1}B^\z,
$$
defined on $S$. For $\re\z=a_0$, inequality \eqref{e51} applied with $p_0=3, q_0=3/2$ instead of $p,q$ yields
\begin{equation}\label{e10}
\begin{split}
\|T_\z\|_{\css_3}
\le C_3 \|A^{2q/3}\|_3\|B^{2q/3}\|_3=C_3 \|A\|^{2q/3}_{2q}\|B\|^{2q/3}_{2q}
\end{split}
\end{equation}
for $\re\z=a_0$. On the other hand, since for $\re\z=0$ we have $|A^{\zeta}|=|B^{\zeta}|=1$ a.e. on $\br^2$, we also see that
\begin{equation}\label{e11}
\|T_{\z}\|_{\css_\infty}
\le \frac {\|\chi\|_{\infty}}{{d(z,\br_+)}}.
\end{equation}
by the spectral theorem for $\dd^2$. Compactness of $T_{\zeta}$ follows by the same argument as in \emph{Case I.1}. Interpolating in between \eqref{e10} and \eqref{e11}, with
$$
\z:=1, \quad 1=\frac{2q}3\cdot\g+0\cdot(1-\g)=\frac{2q}3\,\g,
$$ 
we get $\g=3/(2q)\in (0,1)$ and consequently 
$$
\frac1{p_{0\g}}=\frac\g3+\frac{(1-\g)}\infty=\frac\g3,
$$ 
which means that $p_{0\g}=2q$. That is, 
$$
\|A\chi(D)(\dd^2-z)^{-1}B\|_{\css_{2q}}\le \frac{C_8}{d(z,\br_+)^{1-\g}}\|A\|_{2q}\|B\|_{2q},
$$
By the same limiting argument as before, we get relation  \eqref{e52}.

\medskip\nt
{\bf Case I.1: $1\le q\le 4/3$.} Let $\ti\chi$ be a cutoff function with the same support properties as $\chi$ and such that $\ti\chi=1$ on the support of $\chi$; in particular, $\ti\chi\chi=\chi$. 

Let $A,B\in L^2(\br^2)$. We start by proving that
\begin{equation}\label{e12}
\|A\chi(D)\frac{dE_{\sqrt{-\Delta}}(\l)}{d\lambda}\ti\chi(D)B\|_{\css_1}\le C\, \|A\|_2\|B\|_2.
\end{equation}
Indeed, using \eqref{e11bis}, we re-write the operator on the left hand side of \eqref{e12} as
\begin{align}\label{e12bis}
A\chi(D)\frac{dE_{\sqrt{-\Delta}}(\l)}{d\lambda}\ti\chi(D)B=
\frac{\lambda^{n-1}}{(2\pi)^n}\big (R(\l)\chi(D)A\big)^*\big(R(\l)\ti\chi(D)B\big).
\end{align}
The kernel of the operator $R(\l)\chi(D)A: L^2(\br^2)\to L^2(S_\l)$ is given by
$$
(R(\l)\chi(D)A)(\xi,x)=\chi(\xi)e^{ix\xi}A(x), \quad x\in\br^2,\xi\in S_\l,
$$
and thus
$$
\|R(\l)\chi(D)A\|^2_{\css_2}=\int_{\br^2_x}\int_{S_{\l,\xi}}|\chi(\xi)A(x)|^2\, dxd\s_\sl(\xi)
=\|\chi\|^2_{L^2(\sl)}\|A\|^2_2\leq C\|A\|^2_2.
$$
Since the same bound holds for $R(\l)\ti\chi(D)B$, Hölder's inequality for $\css_p$-classes yields \eqref{e12}.

Set $0<a_0<1$. Using the formula
\begin{align*}
(\dd^2-z)^{-(a_0+it)}=\int_\br (\l^4-z)^{-(a_0+it)}\, dE_{\sqrt{-\Delta}}(\l).
\end{align*}
inequality \eqref{e12} and the fact that the functions $\|\chi_j\|_{\sl}$ are supported on the set where $1/2\leq\l\leq 3/2$, we get the bound
\begin{equation}\label{e121}
\|A\chi(D)(\dd^2-z)^{-(a_0+it)}\chi(D)B\|_{\css_1}\le C\, \frac{e^{2\pi|t|}}{(1-a_0)}\|A\|_2\|B\|_2.
\end{equation}

On the other hand, from \eqref{e501}, we see that 
$$
|\chi(D)(\dd^2-z)^{-3/2+it}(x)|\le C'_1e^{\pi^2 t^2},
$$
that is, the kernel of $\chi(D)(\dd^2-z)^{-3/2+it}(x)$ is uniformly bounded with respect to the ``space variable'' $x\in\br^2$. The Hilbert-Schmidt bound for integral operators implies immediately
\begin{equation}\label{e122}
\|A\chi(D)(\dd^2-z)^{-3/2+it}B\|_{\css_2}\le C\, \|A\|_2\|B\|_2.
\end{equation}
Let $0<\ep<1/2$ be fixed. Suppose, as in \emph{Cases I.2 and II}, that $A,B\in L^2(\br^d)\cap L^\infty_0(\br^2)$. Furthermore, set
$$
T_\z:=A\chi(D)^2(\dd^2-z)^{-\z}B
$$
and $S=S_{a_0,b_0}:=\{\z: a_0\le \re\z \le b_0\}$ to be the vertical strip with 
$$
a_0=\frac{(1-2\ep)}{(1-\ep)}<1,\quad b_0=3/2>1.
$$ 
As previously, the family $(T_\z)$ on $S_{a_0,b_0}$ satisfies the assumptions of Theorem \ref{t43} and we can interpolate between \eqref{e121} and \eqref{e122}. More precisely, 
for the parameters of the corollary we take $\z:=1$ and
$$
1=\frac{1-2\ep}{1-\ep}\g+\frac 32(1-\g),
$$
\ie $\g=(1-\ep)/(1+\ep)$. Hence the relation 
$$
\frac 1{s_\g}=\frac\g 1+\frac{(1-\g)}{2}
$$
gives $s_\g=1+\ep$. To sum up, we arrive at
\begin{equation}\label{e13}
\|A\chi(D)(\dd^2-z)^{-1}B\|_{\css_{1+\ep}}\le C\ep^{-(1-\ep)/(1+\ep)}\|A\|_2\|B\|_2.
\end{equation}
We interpolate once again in between \eqref{e13} and \eqref{e51} for $q=4/3$ to obtain \eqref{e51} for $1\le q<4/3$. Passing from $A,B\in L^{2q}(\br^2)\cap L^\infty_0(\br^2)$ to general $A,B\in L^{2q}(\br^2)$ is carried out as in the previous cases.
\end{proof}

We introduce some notation before going to the proof of Theorem \ref{t1}. Let
$$
k(u)^4:=(u^2-m^2),
$$
where  we use the principal branch of 4-th complex root, so that $k(u)=(u^2-m^2)^{1/4}\in \br_+$ for $u=x\in \br, x>m$. Furthermore,
$$
\z(u):=\frac{u+m}{k(u)^2}=\lp\frac{u+m}{u-m}\rp^{1/2},\ u\not=\pm m
$$
with the standard choice of the branch of the square complex root.

\subsection{Proof of Theorem \ref{t1}}\label{s14}
In order to distinguish the variable refered to in operators $\prt_z,\prt_{\bar z}$ and the spectral parameter of the operator $\bgm$, the latter will be denoted by $u\in\r(\bgm)$ in this subsection.

We consider first \emph{Case I} of the theorem, \ie $1\le q\le 3/2$. Let $A,B\in L^{2q}(\br^2; \mat)$, that is
$$
A(x)=[A_{jl}(x)]_{j,l=1,2}, \quad x=(x_1,x_2)\in\br^2,
$$
and $A_{jl}(x)\in L^{2q}(\br^2)$. Recalling the identities
$$
4\prt_z\prt_{\bar z}=4\prt_{\bar z}\prt_z=\lp \prt^2_{x_1}+\prt^2_{x_2}\rp^2=\dd^2,
$$
we readily see
\begin{eqnarray*}
\bgm^2-u^2&=&
\begin{bmatrix}
m& 4\prt^2_{\bar z}\\
4\prt^2_z& -m
\end{bmatrix}^2-u^2
=
\begin{bmatrix}
\dd^2+(m^2-u^2)& 0\\
0& \dd^2+(m^2-u^2)
\end{bmatrix}\\
&=&(\dd^2-k(u)^4)I_2.
\end{eqnarray*}
For $k(u)^4\in \bc\bsl\br_+$, we have
$$
(\bgm-u)^{-1}=(\dd^2-k(u)^4)^{-1}(\bgm+u).
$$
We are interested in Schatten-von Neumann properties of Birman-Schwinger operator of the bilayer Hamiltionian, \ie
$$
BS_u:=[BS_{u, jl}]_{j,l=1,2}=A(\bgm-u)^{-1}B=A(\dd^2-k(u)^4)^{-1}(\bgm+u)B.
$$
Of course, a bound of the form
$$
\|BS_u\|_{\css_p}\le C(u)\|A\|_{2q}\|B\|_{2q},
$$
see \eqref{bshw1}, \eqref{bshw2}, will follow if we prove it ``entry-by-entry'', that is
$$
\|BS_{u, jl}\|_{\css_p}\le C(u)\|A\|_{2q}\|B\|_{2q}, \quad j,l=1,2.
$$
We shall do the computation for the entry $BS_{u, 11}$; the bounds for other entries of the operator $BS_u$ are obtained in a similar way. We have
\begin{equation}\label{e143}
\begin{split}
&BS_{u, 11}=(m+u)A_{11}(\dd^2-k(u)^4)^{-1}B_{11}+4A_{11}(\dd^2-k(u)^4)^{-1}\prt^2_{\bar z}B_{21}\\
&\quad +4A_{12}(\dd^2-k(u)^4)^{-1}\prt^2_zB_{11}+(m-u)A_{12}(\dd^2-k(u)^4)^{-1}B_{21}.
\end{split}
\end{equation}
To simplify the following computations, we use a homogeneity argument; in detail: let $f\in L^s(\br^2), s>0, \ f=f(x), x\in \br^2$. Set $x=ay, a>0, y\in\br^2$. We write $g(y)=f(ay)$; to make the writing of differential operators more precise, we write $x$- or $y$-subindex to indicate the variable the differential operator is computed with. For instance $\dd_x$ and $\dd_y$ are the Laplacians computed with respect to $x$ and $y$, respectively. 

It is plain that for $j=1,2$
\begin{eqnarray*}
\prt_{y_j}g(y)&=&a\prt_{x_j}f(ay)=a\prt_{x_j}f(x),\\
\prt^2_{y^2_j}g(y)&=&a^2\prt^2_{x^2_j}f(ay)=a^2\prt^2_{x^2_j}f(x).
\end{eqnarray*}
In particular, $\prt_{z,y}g=a\prt_{z,x}f,\ \prt^2_{z,y}g=a^2\prt^2_{z,x}f$, $\dd^2_y g=a^4\dd^2_x f$, etc.

Furthermore, one has 
\begin{equation}\label{e142}
\|g\|^s_s=\int_{\br^2_y}|g(y)|^s\, dy=\int_{\br^2_y}|f(ay)|^s\, dy
=a^{-2}\int_{\br^2_x}|f(x)|^s\, dx=a^{-2}\|f\|^s_s,
\end{equation}
or $\|g\|_s=a^{-2/s}\|f\|_s$.

Suppose that $k(u)\not=0$ and  write $k(u)^4$ as $k(u)^4=|k(u)|^4e^{i\p}$. We assume also that $e^{i\p}\not =1$; the case $e^{i\p}=1$ can be obtained by a standard argument passing to the limit in relations \eqref{e51}, \eqref{e52}. So, putting $a=1/|k(u)|$,
$$
(\dd^2_x-k(u)^4)f(x)=|k(u)|^4(|k(u)|^{-4}\dd^2_x-e^{i\p})f(x)=|k(u)|^4(\dd^2_y-e^{i\p})g(y),
$$
where $g(y)=f(ay), \ x=ay$. In the same way,
$$
\prt^2_{z,x}f(x)=|k(u)|^2\prt^2_{z,y}g(y),\qquad \prt^2_{\bar z,x}f(x)=|k(u)|^2\prt^2_{\bar z,y}g(y).
$$

Set $\ti A_{jl}(y)=A_{jl}(ay)$ and $\ti B_{jl}(y)=B_{jl}(ay)$ for $j,l=1,2$. Turning back to \eqref{e143}, we rewrite it as
\begin{equation} \label{e141}
\begin{split}
&BS_{u, 11}=\frac 1{|k(u)|^2}\Big(\frac{(m+u)}{|k(u)|^2}\ti A_{11}(y)(\dd^2_y-e^{i\p})^{-1}\ti B_{11}(y)+4\ti A_{11}(y)(\dd^2_y-e^{i\p})^{-1}\prt^2_{\bar z,y}\ti B_{21}(y)\\
&\quad +4\ti A_{12}(y)(\dd^2_y-e^{i\p})^{-1}\prt^2_{z,y}\ti B_{11}(y)+\frac{(m-u)}{|k(u)|^2}\ti A_{12}(y)(\dd^2_y-e^{i\p})^{-1}\ti B_{21}(y)\Big).
\end{split}
\end{equation}
Suppose momentarily that we could prove the following estimates,
\begin{eqnarray}\label{e14}
\|\ti A_{11}(\dd^2_y-e^{i\p})^{-1}\ti B_{11}\|_{\css_p}&\le& C \|\ti A_{11}\|_{2q}\|\ti B_{11}\|_{2q},\\
\|\ti A_{11}(\dd^2_y-e^{i\p})^{-1}\prt^2_{\bar z,y}\ti B_{21}\|_{\css_p}&\le&
C \|\ti A_{11}\|_{2q}\|\ti B_{21}\|_{2q}, \nonumber\\
\|\ti A_{12}(\dd^2_y-e^{i\p})^{-1}\prt^2_{z,y}\ti B_{11}\|_{\css_p}&\le&C \|\ti A_{12}\|_{2q}\|\ti B_{11}\|_{2q}, \nonumber\\
\|\ti A_{12}(\dd^2_y-e^{i\p})^{-1}\ti B_{21}\|_{\css_p}&\le &C \|\ti A_{12}\|_{2q}\|\ti B_{21}\|_{2q}, \nonumber
\end{eqnarray}
Recall that $|(m+u)/|k(u)|^2|=|\z(u)|$ and $|(m-u)/|k(u)|^2|=|\z(u)|^{-1}$, while
$$
1\le C(|\z(u)|+|\z(u)|^{-1}),\quad u\in\bc.
$$
Plugging these bounds in \eqref{e141} implies
\begin{eqnarray}\label{e150}
\|BS_{u, 11}\|_{\css_p}&\le&\frac {C}{|k(u)|^2}(1+|\z(u)|+|\z(u)|^{-1}) \|\ti A\|_{2q}\|\ti B\|_{2q}\\
&\le&\frac {C}{|k(u)|^2}(|\z(u)|+|\z(u)|^{-1})\|\ti A\|_{2q}\|\ti B\|_{2q} \nonumber\\
&=&C(|\z(u)|+|\z(u)|^{-1})|k(u)|^{2/q-2}\|A\|_{2q}\|B\|_{2q}, \nonumber
\end{eqnarray}
where we used the rescaling \eqref{e142} in the last line. We notice that
$$
 (|\z(u)|+|\z(u)|^{-1})|k(u)|^{2/q-2}\le C\pp_q(u),\quad u\in\r(\bgm).
$$
Hence \eqref{e150} is exactly the formula claimed in \eqref{bshw1}.

Consequently, it remains to prove \eqref{e14}. Set 
$$
m_1(\xi):=\frac 1{(|\xi|^4-e^{i\p})}, \quad m_2(\xi):=\frac{(\xi_1\pm i\xi_2)^2}{(|\xi|^4-e^{i\p})}.
$$ 
Furthermore, take $\chi_1\in C_0^{\infty}(\br^2)$ with the properties: $0\le \chi_1(x)\le 1$ for all $x\in\br^2$,  $\chi_1$ is supported in $\{x\in\br^2: 1/2\le|x|\le 3/2\}$ and $\chi_1(x)=1$ for $x\in\{x\in \br^2: 3/4\le|x|\le 5/4\}$. Let $\chi_2:=1-\chi_1$; by definition $\chi_1+\chi_2=1$ is a smooth partition of unity.
Rewriting \eqref{e14} in terms of symbols of differential operators, we shall show that
$$
\|\ti A \chi_l(D) m_j(D)\ti B\|_{\css_p}\le C\|\ti A\|_{2q}\|\ti B\|_{2q},\quad l,j=1,2.
$$

For $1\le q\le 3/2$, the bound for $l=1$ is exactly \emph{Case I} of Proposition \ref{p2}. 

Consider the case $l=2$ now. Notice that for the range of $q$'s we are interested in, one can always choose $\ep>0$ small enough so that $p=p(q,\ep)\ge q$. Thus we shall prove the bound
$$
\|\ti A\chi_2(D) m_j(D)\ti B\|_{\css_q}\le C\|\ti A\|_{2q}\|\ti B\|_{2q},\quad j=1,2,
$$
which is stronger than \eqref{e14}. Notice that
\begin{eqnarray*}
|\chi_2(\xi)m_1(\chi)|&=&\big |\frac{\chi_2(\xi)}{|\xi|^4-e^{i\p}}\big|\le\frac C{(1+|\xi|^2)},\\
|\chi_2(\xi)m_2(\chi)|&=&\big |\frac{\chi_2(\xi)(\xi_1\pm i\xi_2)^2}{|\xi|^4-e^{i\p}}\big|\le\frac C{(1+|\xi|^2)}.
\end{eqnarray*}
Lemma \ref{p41} applied to the operator $\ti A\chi_2(D)m_j(D)\ti B$ gives
$$
\|\ti A\chi_2(D) m_j(D)\ti B\|_{\css_q}\le \|(1+|\xi|^2)^{-1}\|_q\, \|\ti A\|_{2q}\|B\|_{2q},\quad j=1,2,
$$
as needed.

\smallskip
Let us turn to \emph{Case II, $q>3/2$}. The proof closely follows the proof of Proposition \ref{p2}, \emph{Case II}. It consists in interpolation in between bounds for para\-me\-ters $q=3/2$ (\ie \emph{Case I}), and $q=\infty$. 

Assume that $A>0$ and $B>0$. Fix $q>3/2$ and let $p=p(q):=2q$. This choice implies in particular that $2q/3>1$. Set $a_0=0, b_0=2q/3$ and consider the strip
$$
S:=\{\z=a+it: a_0\le  a\le b_0,\ t\in\br\}.
$$
The family of operators
$$
T_\z=A^\z(\bgm-u)^{-1}B^\z, 
$$
is analytic on $S$. Apply \eqref{bshw1} with $q_0=3/2$ in place of $q$ to the family $T_\z$ on $\re\z=b_0=2q/3$; that is
\begin{equation}\label{e16}
\|A^{2q/3+it}(\bgm-u)^{-1}B^{2q/3+it}\|_{\css_3}\le C\pp(u)\|A\|^{2q/3}_{2q}\|B\|^{2q/3}_{2q},
\end{equation}
where we used that  $\|A^{2q/3+it}\|_3=\|A\|^{2q/3}_{2q}$, and the same relation holds for $B$. Notice that $q_{01}=1-1/(2q_0)=2/3$. For $\re\z=a_0=0$, we have the trivial bound
\begin{equation}\label{e161}
\|A^{it}(\bgm-u)^{-1}B^{it}\|_{\css_\infty}\le \frac 1{d(u,\s(\bgm))}.
\end{equation}
As in Proposition \ref{p2}, we interpolate between \eqref{e16} and \eqref{e161} using Theorem \ref{t43} with parameters $\z:=1$ and
$$
1=\frac{2q}{3}\, \g+(1-\g)\, 0,\quad \frac 1{p_\g}=\frac 3\g+\frac{(1-\g)}\infty=\frac 1{2q}.
$$
Hence, $\g=3/(2q)$ and $p_\g=2q$. Claim \eqref{bshw2} follows, and this finishes the proof of the theorem.
\hfill $\Box$

\section{Lieb--Thirring inequalities for bilayer graphene}\label{s2}
In what follows we always assume that $m>0$. We begin with the standard Zhukovsky transform
\begin{equation}\label{zhuk}
z=z(w)=\frac{m}2\,\Bigl(w+\frac1{w}\Bigr),
\end{equation}
which maps the upper half-plane $\bc_+$ onto the domain $\rho(\bgm)$. Since
$$ |z(w)\pm m|=\frac{m}{2|w|}\,|w\pm 1|^2, $$
we have
\begin{equation}\label{zhuk1}
\begin{split}
|z+m|+|z-m| &=\frac{m}{2|w|}\,\Bigl(|w+1|^2+|w-1|^2\Bigr)=\frac{m}{|w|}\,\bigl(1+|w|^2\bigr), \\
|z^2-m^2|^{\frac12} &=\frac{m}{2|w|}\,|w^2-1|.
\end{split}
\end{equation}
The distortion \cite[Cor. 1.4]{pom} for the Zhukovsky transform reads as
\begin{equation}\label{dist}
\frac{d(z,\s(\bgm))}{\im w}\asymp |z'(w)|=\frac{m|w^2-1|}{2|w|^2}=\frac{|z^2-m^2|^{1/2}}{|w|}\,, \qquad w\in\bc_+.
\end{equation}

\subsection{Proof of Theorem \ref{t-lthin1}, Case I: $1<q\le 3/2$}\label{s21}
We have, by \eqref{zhuk1},
$$ \Phi(z(w))=C(1+|w|^2)\,\frac{|w|^{p_1}}{|w^2-1|^{2q_1}}\,, \qquad p_1:=2q_1-1=1-\frac1q>0. $$
The bound \eqref{bshw1} in the variable $w$ reads
\begin{equation}\label{bshw11}
\|V_2(\bgm-z(w))^{-1}V_1\|_{\css_{p}}\le C_9(1+|w|^2)\,\frac{|w|^{p_1}}{|w^2-1|^{2q_1}}\,\|V\|_q, \quad w\in\bc_+,
\end{equation}
where $V_2=A:=|V|^{1/2}$ and $V_1=B:=V^{1/2}$, see the discussion preceding \eqref{e004}.
For $w=iy$, $y>0$,
\begin{equation}\label{bshw12}
\|V_2(\bgm-z(iy))^{-1}V_1\|_{\css_{p}}\le C_9\left(\frac{y}{1+y^2}\right)^{p_1}\,\|V\|_q<\frac{C_9}{y^{p_1}}\,\|V\|_q.
\end{equation}
We proceed with the {\it regularized perturbation determinant}
$$ H(w):=\det_p \bigl(I+V_2(\bgm-z(w))^{-1}V_1\bigr), \qquad w\in\bc_+, $$
which admits the bounds, see \cite[Thm. 9.2]{si1}
\begin{equation}\label{perdet1}
\log |H(w)| \le \gga_p\,\|V_2(\bgm-z(w))^{-1}V_1\|_{\css_{p}}^p
\end{equation}
and
\begin{equation}\label{perdet2}
|H(w)-1|\le\p\bigl(\|V_2(\bgm-z(w))^{-1}V_1\|_{\css_{p}}\bigr), 
\end{equation}
where
$$
\p(x):=x\exp\{\gga_p(x+1)^p\}, \quad x\ge0.
$$
Denote
\begin{equation}\label{norm} 
h(w)=h_y(w):=\frac{H(yw)}{H(iy)}\,, \qquad h(i)=1, 
\end{equation}
$y\ge1$ is chosen later on.

\begin{proposition}\label{pr1}
Assume that
\begin{equation}\label{perboun}
\|V\|_q\le 1.
\end{equation}
Then there is a constant $C_{10}=C_{10}(m,q,\ep)$ so that for $y=C_{10}$ the following holds
\begin{equation*}
\log |h(w)|\le C_{11}\,\frac{(1+|w|)^{4pq_1}}{|w^2-y^{-2}|^{2pq_1}}\,\|V\|_q, \qquad w\in\bc_+.
\end{equation*}
\end{proposition}
\begin{proof}
Without loss of generality we assume that $C_9>1$. If $y^{p_1}\ge C_9\ge C_9\|V\|_q$, we have, by \eqref{bshw12},
\begin{equation}\label{bshw13}
\|V_2(\bgm-z(iy))^{-1}V_1\|_{\css_{p}}\le\frac{C_9}{y^{p_1}}\,\|V\|_q\le\|V\|_q\le 1.
\end{equation}
An obvious bound $\p(x)\le\exp\{2^p\gga_p\}\,x$, $0\le x\le 1$, implies, in view of \eqref{bshw13},
$$ \p\bigl(\|V_2(\bgm-z(iy))^{-1}V_1\|_{\css_{p}}\bigr)\le e^{2^p\gga_p}\,\|V_2(\bgm-z(iy))^{-1}V_1\|_{\css_{p}}, $$
and so, by \eqref{perdet2},
$$ 
1-|H(iy)|\le |H(iy)-1|\le \frac{C_9e^{2^p\gga_p}}{y^{p_1}}\,\|V\|_q\le \frac12, 
$$
as soon as $y^{p_1}\ge 2C_9\exp\{2^p\gga_p\}=:C_{12}$. The case $|H(iy)|>1$ being trivial, we continue with the case $\frac 12\le |H(iy)|\le 1$. Hence,
\begin{equation}\label{perdet3}
|H(iy)|\ge\frac12\,, \quad \log|H(iy)|\ge -2\bigl(1-|H(iy)|\bigr)\ge -C_{12}\,\frac{\|V\|_q}{y^{p_1}}\,.
\end{equation}

A combination of \eqref{perdet1}, \eqref{bshw11}, and \eqref{perdet3} leads to the bound
\begin{equation*}
\begin{split}
\log|h(w)| &= \log|H(yw)|-\log|H(iy)| \\
&\le C(1+y|w|)^{2p}\frac{(y|w|)^{pp_1}}{|y^2w^2-1|^{2pq_1}}\,\|V\|_q^p+C_{12}\,\frac{\|V\|_q}{y^{p_1}} \\
&\le C_{13}\left[\frac{(1+|w|)^{2p}|w|^{pp_1}}{|w^2-y^{-2}|^{2pq_1}}\,\frac{\|V\|_q^p}{y^{pp_1}}+\frac{\|V\|_q}{y^{p_1}}\right] \\
&\le C_{13}\frac{\|V\|_q}{y^{p_1}}\left[\frac{(1+|w|)^{2p}|w|^{pp_1}}{|w^2-y^{-2}|^{2pq_1}}+1\right].
\end{split}
\end{equation*}
As $2p+pp_1-4pq_1=-pp_1<0$, we have for $y\ge1$
$$ (1+|w|)^{2p}|w|^{pp_1}+|w^2-y^{-2}|^{2pq_1}\le (1+|w|)^{2p+pp_1}+(1+|w|)^{4pq_1}<2(1+|w|)^{4pq_1}. $$
The result follows with $y=C_{10}=C_{12}^{1/p_1}$, $C_{11}=2C_{13}$.
\end{proof}

It is well known that the Lieb--Thirring inequalities agree with the Blaschke type conditions for the zeros of the
corresponding perturbation determinants. So, the next step is an application of \cite[Thm.~4.4]{bgk2} to the above function $h$.
The input parameters are
\begin{equation*}
\begin{split}
a &=0, \quad b=2pq_1, \quad c_j=0; \quad x_1'=y^{-1}, \ x_2'=-y^{-1}, \quad K=C\|V\|_q, \\
d_1 &=d_2=d=2pq_1=\left\{
                  \begin{array}{ll}
                    \frac{2q-1}{2-q}+(2-\frac1q)\ep, & 1<q<\frac43; \\
                    \frac{2q-1}{2-q}, & \frac43\le q\le\frac32.
                  \end{array}
                \right.
\end{split}
\end{equation*}
The output parameters in \cite[Thm.~4.4]{bgk2} are
$$ l=\{l\}_{a,\ep}=0, \quad (d-1+\ep)_+=\frac{3q-3}{2-q}+\o_q\ep, \quad l_1=\frac{4q-2}{2-q}+\tau_q\ep, $$
with
$$ \o_q=\left\{
                  \begin{array}{cc}
                    \frac{3q-1}{q}, & 1<q<\frac43; \\
                    1, & \frac43\le q\le\frac32.
                  \end{array}
                \right. \quad
   \tau_q=\left\{
                  \begin{array}{cc}
                    \frac{6q-1}{q}, & 1<q<\frac43; \\
                    1, & \frac43\le q\le\frac32.
                  \end{array}
                \right.
$$
So, the Blaschke type condition of \cite[Thm.~4.4]{bgk2}
takes the form
\begin{equation}\label{bla1}
\sum_{\xi\in Z(h)} \frac{(\im\xi)^{1+\ep}}{(1+|\xi|)^{l_1}}\,|\xi^2-y^{-2}|^{(d-1+\ep)_+}\le C_{14}\|V\|_q,
\end{equation}
and, since the ``test point'' $y$ in Proposition \ref{pr1} does not depend on $V$, the constant $C_{14}(m,q,\ep)$ does not depend on $V$ either.

In terms of the zeros of $H$ we have
$$ \xi\in Z(h) \ \Leftrightarrow \ y\,\xi=\l\in Z(H), \qquad \xi=\frac{\l}{y}\,, $$
and as $y=C_{10}$ is a constant, condition \eqref{bla1} does not alter
\begin{equation}\label{bla2}
\sum_{\l\in Z(H)} \frac{(\im\l)^{1+\ep}}{(1+|\l|)^{l_1}}\,|\l^2-1|^{(d-1+\ep)_+}\le C_{15}\|V\|_q.
\end{equation}

It remains to get back to the spectral variable $z\in\rho(\bgm)$, keeping in mind that for the discrete spectrum of $\bg$
the equivalence holds
$$ \z\in\s_d(\bg) \ \Leftrightarrow \ \l\in Z(H). $$
To make the final result transparent, we invoke the main result \cite[Theorem 1.1]{jcc2}, which claims, in particular,
that the discrete spectrum $\s_d(\bg)$ is bounded, that is, $|\z|\le C_{16}$, $\forall\z\in\s_d(\bg)$.  In the Zhukovsky variable the latter means
\begin{equation}\label{bouzhu}
 0<c\le |\l|\le C<\infty, \qquad \forall\l\in Z(H).
\end{equation}
So we can neglect the term $1+|\l|$ in \eqref{bla1}. Next, as in \eqref{zhuk1},
$$ |\z^2-m^2|=\frac{m^2}4\,\frac{|\l^2-1|^2}{|\l|^2} \ \Rightarrow \ c|\l^2-1|\le |\z^2-m^2|^{1/2}\le C|\l^2-1|. $$
Finally, the distortions \eqref{dist} and \eqref{bouzhu} imply
$$ c\,\im\l \le \frac{d(\z,\s(\bgm))}{|\z^2-m^2|^{1/2}}\le C\,\im\l. 
$$
\emph{Case I} of Theorem \ref{t-lthin1} is proved. \hfill $\Box$

\subsection{Proof of Theorem \ref{t-lthin1}, Case II: $q>3/2$}\label{s22}
We use the distortion \eqref{dist} to obtain the bound similar to \eqref{bshw11}
\begin{equation}\label{bshw21}
\|V_2(\bgm-z(w))^{-1}V_1\|_{\css_{p}}\le C_9\frac{(1+|w|)^{2q_2}}{(\im w)^{p_2}}\,\frac{|w|^{p_3}}{|w^2-1|^{p_4}}\,\|V\|_q,
\quad w\in\bc_+,
\end{equation}
where
$$ p=2q, \quad p_2:=1-q_2=1-\frac3{2q}>0, \quad p_3:=2-\frac5{2q}\,, \quad p_4:=1+\frac1{2q}\,. $$
Note that $p_3-p_2=p_1$. For $w=iy$, $y>0$, the bound is exactly the same as \eqref{bshw12}
\begin{equation}\label{bshw22}
\|V_2(\bgm-z(iy))^{-1}V_1\|_{\css_{p}}<\frac{C_9}{y^{p_1}}\,\|V\|_q.
\end{equation}
We argue as in the proof of Proposition \ref{pr1} to obtain the bound for $h$ \eqref{norm}
\begin{equation}\label{normboun}
\log|h(w)|\le C_{11}\,\frac{|w|^{pp_2}(1+|w|)^{2pp_4}}{(\im w)^{pp_2}\,|w^2-y^{-2}|^{pp_4}}\,\|V\|_q.
\end{equation}
Indeed,
\begin{equation*}
\begin{split}
\log|h(w)| &= \log|H(yw)|-\log|H(iy)| \\
&\le C\frac{(1+y|w|)^{2pq_2}(y|w|)^{pp_3}}{(\im yw)^{pp_2}|y^2w^2-1|^{pp_4}}\,\|V\|_q^p+C_{12}\,\frac{\|V\|_q}{y^{p_1}} \\
&\le C_{13}\left[\frac{(1+|w|)^{2pq_2}|w|^{pp_3}}{(\im w)^{pp_2}|w^2-y^{-2}|^{pp_4}}\,\frac{\|V\|_q^p}{y^{pp_1}}+\frac{\|V\|_q}{y^{p_1}}\right] \\
&\le C_{13}\frac{\|V\|_q}{y^{p_1}}\left[\frac{(1+|w|)^{2pq_2}|w|^{pp_3}}{(\im w)^{pp_2}|w^2-y^{-2}|^{pp_4}}+1\right].
\end{split}
\end{equation*}
Next,
\begin{equation*}
\begin{split}
&{} (1+|w|)^{2pq_2}|w|^{pp_3}+(\im w)^{pp_2}|w^2-y^{-2}|^{pp_4}\le (1+|w|)^{2pq_2}|w|^{pp_3}+|w|^{pp_2}(1+|w|^2)^{pp_4} \\
&\le |w|^{pp_2}\Bigl((1+|w|)^{2pq_2}|w|^{pp_1}+(1+|w|)^{2pp_4}\Bigr)\le 2|w|^{pp_2}(1+|w|)^{2pp_4},
\end{split}
\end{equation*}
and \eqref{normboun} follows.

The computation with \cite[Thm.~4.4]{bgk2} is a bit more complicated now. The input parameters are
\begin{equation*}
\begin{split}
a &=pp_2=2q-3>0, \quad b=pp_4=2q+1, \quad x_1'=y^{-1}, \quad x_2'=-y^{-1}, \quad x_1=0, \\
c_1 &=pp_2=a, \quad c_j=0, \ j\ge 2, \quad d_1=d_2=d=pp_4=b, \quad K=C\|V\|_q.
\end{split}
\end{equation*}
The output parameters in \cite[Thm.~4.4]{bgk2} are
$$ l=a, \quad \{l\}_{a,\ep}=-a, \quad (d-1+\ep)_+=2q+\ep, \quad l_1=2+4q+4\ep, $$
so the Blaschke type condition takes the form
$$ \sum_{\xi\in Z(h)} \frac{(\im \xi)^{a+1+\ep}}{(1+|\xi|)^{2+4q+4\ep}}\,\frac{|\xi^2-y^{-2}|^{2q+\ep}}{|\xi|^a}\le C_{14}\|V\|_q. $$
After the change of variable $\l=y\,\xi=C_{10}\xi$, we come to
\begin{equation}\label{bla3} 
\sum_{\l\in Z(H)} \frac{(\im \l)^{a+1+\ep}}{(1+|\l|)^{2+4q+4\ep}}\,\frac{|\l^2-1|^{2q+\ep}}{|\l|^a}\le C_{15}\|V\|_q. 
\end{equation}

As before, the final step relies on the distortion relations for the Zhukovsky transform. Indeed,  separate the upper-half plane $\bc_+$ in three regions $\oo_1:=\{\l\in \bc_+: c\le|\l|\le C\}, \ \oo_2:=\{\l\in\bc_+: |\l|\ge C\}$ and $\oo_3:=\{\l\in\bc_+: |\l|\le c\}$ with constants $c, C$ chosen as $0<c<1<C<+\infty$. 
It is clear that
$$
\sum_{\l\in Z(H)\cap \oo_1} (\im \l)^{a+1+\ep}\, |\l^2-1|^{2q+\ep}\le C
\sum_{\l\in Z(H)\cap \oo_1} \frac{(\im \l)^{a+1+\ep}}{(1+|\l|)^{2+4q+4\ep}}\,\frac{|\l^2-1|^{2q+\ep}}{|\l|^a}.
$$
On the other hand,  one has $|\zeta(\l)|\asymp|\l|$ for $\l\in \oo_2$, and 
$|\zeta(\l)|\asymp |\l|^{-1}$ for $\l\in \oo_3$. Using these relations along with inequalities given next to \eqref{bouzhu}, we cut the sum \eqref{bla3} in parts corresponding to domains $\oo_i,\ i=1,2,3$, and rewrite these partial sums in terms of $\zeta$-variable.

\emph{Case II} of Theorem \ref{t-lthin1} is proved as well.\hfill $\Box$

\section{Some technical tools: interpolation theorems and Kato-Selier-Simon lemma}\label{s4}
\subsection{Kato-Selier-Simon lemma}\label{s41}
Recall the notation introduced in Section \ref{s11}. We have the following proposition usually called Kato-Selier-Simon lemma. 
\begin{proposition}[{\cite[Thm. 4.1]{si1}}]\label{p41}\hfill
\begin{enumerate}
\item Let $f, g\in L^q(\br^d),\ d\ge 1$. Then, for $2\le q<\infty$, $f(x)g(D)\in \css_q$, and
$$
\|f(x)g(D)\|_{\css_q}\le (2\pi)^{-d}\|f\|_q\|g\|_q.
$$
\item Let $f\in L^q(\br^d),\ d\ge 1$, and $A,B\in L^{2q}(\br^d)$. For $2\le q<\infty$,
$$
\|A(x) f(D) B(y)\|_{\css_q}\le (2\pi)^{-d}\|f\|_q\, \|A\|_{2q}\|B\|_{2q}.
$$
\end{enumerate}
\end{proposition}
The first claim of the above proposition is in Simon \cite[Thm. 4.1]{si1}; the second claim is a ``symmetrized'' version of the first one and it is proved similarly.

\subsection{Interpolation theorem for bounded analytic families}\label{s42}
In this subsection, we follow mainly the presentation of Zhu \cite[Ch. 2]{zhu}.

Let $X_0, X_1$ be two Banach spaces. We say that the pair $X_0, X_1$ is compatible, if there is a topological Hausdorff space $X$ containing both $X_0$ and $X_1$. We have the following theorem.
\begin{theorem}[{\cite[Thm. 2.4]{zhu}}]\label{t42}
Let $X_0, X_1$ be a pair of compatible Banach spaces, idem for $Y_0,Y_1$. For a $\g,\ 0<\g<1$, there are Banach spaces $X_\g, Y_\g$,
$$
X_\g=[X_0,X_1]_\g, \qquad Y_\g=[Y_0,Y_1]_\g,
$$
interpolating in between $X_0$ and $X_1$ and   $Y_0$ and $Y_1$, respectively, in the following sense.

Let $T:X_0+X_1\to Y_0+Y_1$ be a \emph{bounded} linear map such that
\begin{eqnarray*}
\|Tx\|_{Y_0}&\le& C_0\|x\|_{X_0}, \quad x\in X_0,\\
\|Tx\|_{Y_1}&\le& C_1\|x\|_{X_1}, \quad x\in X_1.
\end{eqnarray*}
Then $T$ induces a linear map $T_\g:X_\g\to Y_\g$ with the property
$$
\|T_\g\|\le C_0^\g C^{1-\g}_1.
$$
\end{theorem}
Saying ``interpolation'' we mean ``complex interpolation'' throughout the article. For instance, we have
\begin{equation}\label{e1001}
[L^{p_0}(\br^d),L^{p_1}(\br^d)]_\g=L^{p_\g}(\br^d),
\end{equation}
where $1\le p_0, p_1\le \infty$, $1/p_\g=\g/p_0+(1-\g)/p_1$, and $d\ge 1$, see \cite[Thm. 2.5]{zhu}. 

It is important that a similar construction holds for ``non-commutative'' $L^p$-spaces as well. That is, denoting by $\css_p$ the Schatten-von Neumann classes of compact operators, we have
$$
[\css_{p_0}, \css_{p_1}]_\g=\css_{p_\g},
$$
where $1\le p_0, p_1\le \infty$ and $1/p_\g=\g/p_0+(1-\g)/p_1$. A proof of this result is in \cite[Thm. 2.6]{zhu}. Much more information and further references on the interpolation theory of Banach spaces are in monographs Bennett-Sharpley \cite{bsh} and Bergh-Löfström \cite{bl}.

For  $1\le p_{01}, p_{02}\le +\infty$, it is plain to see that
$$
L^{p_{01}}(\br^d_x)\times L^{p_{02}}(\br^d_y) \simeq
L^{p_{01}}(\br^d_x) \dotplus L^{p_{02}}(\br^d_y), \quad x,y\in \br^d,
$$
and so interpolation \eqref{e1001} holds for these spaces as well. 
This observation is often applied to an operator $\ca$ of the form 
$$
\ca: L^{p_{01}}(\br^d)\times L^{p_{02}}(\br^d)\to \css_{q_{01}}, \quad 1\le q_{01}\le +\infty,
$$
see Section \ref{s1}. 

\subsection{Interpolation theorem for general analytic families}
Following Gohberg-Krein \cite[Ch. III.13]{gk1}, we present a generalized version of interpolation in between $\css_p$-spaces.

Let $a,b\in \br, \ a<b$ and 
$$
S=\{\z: a\le\re\z\le b\}
$$
be a vertical strip in the complex plane. For a Hilbert space $H$, we say that a family of bounded operators $(T_\z)_{\z\in S}, T_\z: H\to H$ is is \emph{analytic} on $S$, if $(T_\z f,g)$ is analytic on an open neighborhood of $S$ for any fixed $f,g\in H$. 

\begin{theorem}[{\cite[Thm. 13.1]{gk1}}]\label{t43}
Let $(T_\z)_{\z\in S}$ be an analytic family of operators. Assume that for any $f,g\in H$
\begin{equation}\label{e20}
\log |(T_\z f,g)|\le C_{1; f,g}e^{C_{2; f,g}|\im\z|},\quad \z\in S,
\end{equation}
where the constants $C_{j; f,g},\ j=1,2$ depend on $f,g$, but not on $\z\in S$, and
$$
0\le C_{2; f,g}<\frac\pi{(b-a)}.
$$

Furthermore, suppose that
\begin{enumerate}
\item for $\re\z=a$, $T_\z\in \css_{p_0}$, with $1\le p_0\le \infty$ and
$$
\|T_\z\|_{\css_{p_0}}\le C_0.
$$
\item  for $\re\z=b$, $T_\z\in \css_{p_1}$, with $1\le p_1<p_0$ and
$$
\|T_\z\|_{\css_{p_1}}\le C_1.
$$
\end{enumerate}

Take an \ $x\in (a,b)$ and write it as $x=\g\, a+(1-\g)\, b, \ \g\in (0,1)$. For $\z\in S, \re\z=x$ we have that $T_\z\in \css_{p_\g}$, and moreover
$$
\|T_\z\|_{\css_{p_\g}}\le C_0^\g\, C_1^{1-\g},
$$
where $1/p_\g=\g/p_0+(1-\g)/p_1$.
\end{theorem}

We often use the following corollary of the above theorem.
\begin{corollary}\label{c01} Let $(T_\z)_{\z\in S}$ be an analytic family of operators satisfying the assumption of Theorem \ref{t43} with conditions (1), (2) replaced by the following assumptions:
\begin{enumerate}
\item[(1')]  for $\re\z=a$, $T_\z\in \css_{p_0}$, with $1\le p_0\le \infty$ and
$$
\|T_\z\|_{\css_{p_0}}\le C_0e^{A_0|\im \z|^2}.
$$
\item[(2')]  for $\re\z=b$, $T_\z\in \css_{p_1}$, with $1\le p_1<p_0$ and
$$
\|T_\z\|_{\css_{p_1}}\le C_1e^{A_1|\im \z|^2},
$$
for some constants $A_0, A_1\ge 0$. 
\end{enumerate}
As above, for an $x=\g\, a+(1-\g)\, b\in (a,b), \ \g\in (0,1)$ and $\z\in S, \re\z=x$ we have that $T_\z\in \css_{p_\g}$, and moreover
$$
\|T_x\|_{\css_{p_\g}}\le C''\, C_0^\g\, C_1^{1-\g},
$$
where $1/p_\g=\g/p_0+(1-\g)/p_1$. The constant $C''$ depends on $a,b, C_0, C_1, A_0$ and $A_1$. 
\end{corollary}

The corollary follows immediately by applying Theorem \ref{t43} to the analytic family
of operators $\ti T_\z=e^{\max(A_0,A_1)\z^2}\, T_\z, \ \z\in S$.

\medskip\nt
{\bf Acknowledgments.} The third and the fourth authors were partially supported by the grant ANR-18-CE40-0035. A part of this research was done during a visit of the third author to the Institute of Mathematics of Bordeaux (IMB UMR5251) of University of Bordeaux. He is grateful to the institution for the hospitality.

\end{document}